\documentclass[letterpaper, 10 pt, conference]{ieeeconf}
\IEEEoverridecommandlockouts
\usepackage{amsmath,amssymb,amsfonts}
\usepackage{commath}
\usepackage{graphicx}
\usepackage{algorithm}
\usepackage{mathtools}
\usepackage{algpseudocode}  
\usepackage{tikz}
\usepackage{comment}
\usepackage{subcaption}
\usepackage{booktabs} 
\newtheorem{theorem}{Theorem}
\newtheorem{assumption}{Assumption}
\newtheorem{proposition}{Proposition}
\newtheorem{problem}{Problem}

\usepackage{multirow}
\usepackage{siunitx}
\title{\LARGE \bf Bandwidth reduction methods for packetized MPC over lossy networks}
\author{Alberto Mingoia$^{1,2}$ \and Matthias Pezzutto$^{1}$ \and Fernando S Barbosa$^{2}$ \and David Umsonst$^{2}$
\thanks{$^{1}$Department of Information Engineering, University of Padova, Italy, 
        {\tt\small alberto.mingoia@studenti.unipd.it, matthias.pezzutto@unipd.it}}%
\thanks{$^{2}$Ericsson Research, Sweden, 
        {\tt\small \{alberto.mingoia, fernando.dos.santos.barbosa, david.umsonst\}@ericsson.com}}%
}

\begin{document}
\maketitle
\begin{abstract}
    We study the design of an offloaded model predictive control (MPC) operating over a lossy communication channel. 
    We introduce a controller design that utilizes two complementary bandwidth-reduction methods.
    The first method is a multi-horizon MPC formulation that decreases the number of  optimization variables, and therefore the size of transmitted input trajectories.
    The second method is a communication-rate reduction mechanism that lowers the frequency of packet transmissions.
    We derive theoretical guarantees on recursive feasibility and constraint satisfaction under minimal assumptions on packet loss, and we establish reference-tracking performance for the rate-reduction strategy.
    The proposed methods are validated using a hardware‑in‑the‑loop setup with a real 5G network, demonstrating simultaneous improvements in bandwidth efficiency and computational load.
\end{abstract}

\section{Introduction}
As 5G technology matures and standardization of 6G technology is ramping up, large-scale wireless control applications are becoming feasible in domains such as smart cities \cite{Smart_Cities} and Industry 4.0 \cite{Industry_4.0}. In these settings, multiple agents share a constrained and imperfect wireless medium, and packet drops, delays, and bandwidth limits become critical to closed-loop control performance. Consequently, the controller design must account for communication constraints \cite{PEZZUTTO2024100972}.

Model predictive control (MPC) is a natural fit for wireless control because it provides a trajectory of optimal inputs, which, if sent in its entirety, can enable local buffering and graceful operation during packet losses \cite{QUEVEDO20121803}. However, sending long trajectories increases communication load, exposing a trade-off between bandwidth consumption and closed-loop performance and robustness. Existing approaches often pose restrictive assumptions on network models, such as bounded delays \cite{Bounded_delays} and having a perfect acknowledgment mechanism \cite{Perfect_ACK}, or do not take into account communication and computational overhead \cite{Remote_MPC,Tube_MPC}.

This work aims to reduce the bandwidth utilized by communication while maintaining performance and safety under minimal network assumptions. We propose two complementary strategies: (i) a transmission-rate reduction policy that adapts number of packets sent between the plant and the controller, and (ii) a reduction in input-trajectory size through multi-horizon MPC \cite{MH_MPC} that exploits models of different granularity, allowing to reduce the size of control packets while also lowering the solve time of the MPC problem. We propose a terminal set design that guarantees constraint satisfaction and recursive feasibility over unreliable channels. We prove that convergence to any admissible constant reference can also be guaranteed in the case where the horizon of the MPC problem is chosen to be uniform. The two strategies lead to simultaneous gains in bandwidth efficiency and computational time. The control solution is implemented and validated with simulations-in-the-loop over a real 5G network.

The remainder of the paper is organized as follows. In Section~\ref{sec:problem formulation} we present the problem formally, while Section~\ref{sec:preliminaries} introduces the necessary preliminaries to follow the paper, namely standard and multi-horizon MPC and tracking of steady states. In Section~\ref{sec:multi_horizon_mpc} a recursively feasible multi-horizon MPC for tracking constant references is derived. In Section~\ref{sec:bw_aware_control} we present the novel control algorithm. In Section~\ref{sec:Experiments} the results are evaluated over a real 5G network with a hardware in the loop setup.

\textit{Notation}:
We denote an $n\times m$ matrix of zeros as $\mathbf{0}_{n\times m}$, the $n\times n$ identity matrix as $I_{n}$, and the modulo (or remainder) operator with $\mathrm{mod}$. The Euclidean norm operator is denoted with $\norm{\cdot}$, and ${\norm{v}_{P}= v^TPv}$, where $v^T$ represents the transposed column vector $v$. Let ${x\in\mathbb{R}^n}$ and $u\in\mathbb{R}^m$ be real-valued column vectors, then $(x,u)\in\mathbb{R}^{n+m}$ represents the stacked column vector. Given a polytope ${\mathcal{P} = \{(x,u) \in \mathbb{R}^{n+m} | P_xx+P_u u \leq P_c \}}$ then its projection on the $x$ space is given by ${Proj_x(\mathcal{P}) = \{x\in \mathbb{R}^n \ | \ \exists u \in \mathbb{R}^m \ : \ P_xx+P_u u \leq P_c \}}$.

\section{Problem Definition}\label{sec:problem formulation}
The closed-loop system is  divided into local (smart actuator and plant) and cloud (estimator and controller) components (see Fig.~\ref{fig:setup}). Local and cloud components are connected through a wireless network.
\begin{figure}
    \centering
    \includegraphics[width=\linewidth]{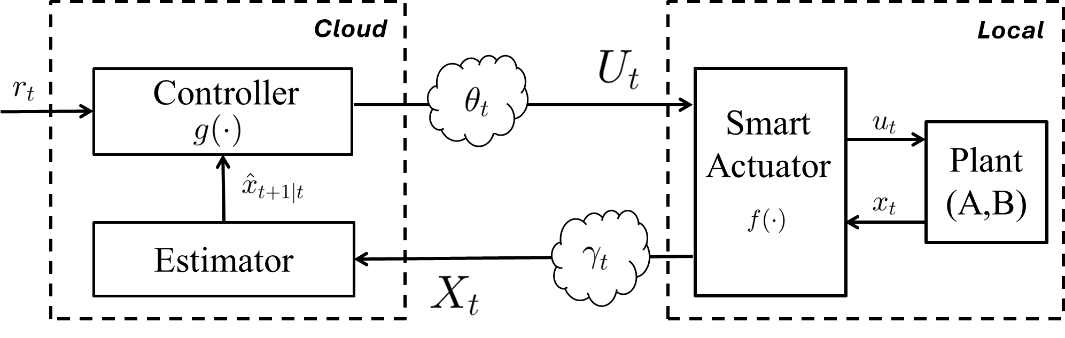}
    \caption{Components of the networked control system.}
    \label{fig:setup}
\end{figure}
On the local side, we consider a discrete-time linear time-invariant plant
\begin{equation}\label{eq:sys}
    x_{t+1} = Ax_t + Bu_t
\end{equation}
with state $x_t \in \mathbb{R}^{n_x}$ and input $u_t \in \mathbb{R}^m$. 
The system is subject to a set of constraints, $x_t\in\mathcal{X}$ and $u_t\in\mathcal{U}$, where
\begin{equation}
\label{eq:constraints}
  \mathcal{X} = \left\{ x\in\mathbb{R}^{n_x} : P_x x \le p_x \right\},\ 
  \mathcal{U} = \left\{ u : P_u u \le p_u \right\},
\end{equation}
with $P_x \in \mathbb{R}^{c_x\times n_x}, p_x \in \mathbb{R}^{c_x}$, $P_u \in \mathbb{R}^{c_u\times m}$ and $p_u\in \mathbb{R}^{c_u}$ where $c_x$, $c_u$ are the number of state and input constraints.

The smart actuator is endowed with limited computational power and determines the input $u_t$ to be applied to the plant at each time-step based on packets received from the cloud and the state of the plant. 
The smart actuator resides next to the plant and therefore directly measures its state.
Furthermore, it timestamps and transmits packets to the cloud. 
On the cloud side, an estimator estimates the state of the plant based on received packets, while a controller uses the estimates and the reference $r_t$ to determine the control inputs sent to the local side.

 The remote controller sends packets $U_t$ during the period $(t-1,t)$, while the local side sends packets $X_t$ during ${(t,t+1)}$. The network is modeled by random binary variables $\theta_t$ and $\gamma_t$.  The variable $\theta_t \in \{0,1\}$ is equal to one if the packet $U_t$ is available at the local side at time $t$, and zero otherwise. The variable $\gamma_t \in \{0,1\}$ is equal to one if $X_t$ is available on the cloud side to compute $U_{t+1}$, and zero otherwise. Note that these random variables do \textit{not specify why} a packet is unavailable at time $t$; the packet may have been dropped or arrived too late due to network latency or high processing time. The following assumption on the network is made.
\begin{assumption} \label{ass:network}
Over an infinite period of time, there exists an infinite number of two successful consecutive transmissions from plant to controller and controller to plant, i.e. $\forall~t~\in~\mathbb{Z},\ \exists\bar t  \geq t \text{ s.t. } \gamma_{\bar t}=1 \ \text{and} \ \theta_{\bar t+1}=1 $.
\end{assumption}

Note that Assumption \ref{ass:network} does not include any restrictions on the delay or packet loss distributions and is, therefore, a very mild assumption on the network. The problem addressed in this paper can be formulated as follows.

\begin{problem}\label{problem}
    Design a control law and smart actuator law to track a reference signal $r_t$ while enforcing the constraints \eqref{eq:constraints}, communicating over a network fulfilling Assumption \ref{ass:network}.
\end{problem}
\section{Preliminaries}\label{sec:preliminaries}
As the proposed solution is based on MPC for tracking \cite{MPC_for_tracking} and multi-horizon MPC \cite{MH_MPC}, a summary of the required background is presented in this section.

\subsection{Model Predictive Control}
In MPC, given a state measurement $x_t$, the optimal control problem (OCP) 
\begin{equation}
\label{eq:general_MPC_formulation}
\begin{aligned}
    &\min_{\substack{\mathbf{u}_t,\mathbf{x}_t}} \sum_{k=0}^{N-1}\Big( \lVert x_{(t,k)}\rVert_{Q}^2 +\lVert u_{(t,k)} \rVert_{R}^2 \Big)+ \lVert x_{(t,N)}\rVert_{P}^2 \\
     \text{s.t.} \quad
&x_{(t,k+1)} = A x_{(t,k)} + B u_{(t,k)}, \quad   \\
& x_{(t,k)} \in \mathcal{X},\ u_{(t,k)} \in \mathcal{U}, \quad \forall k \in \mathbb{Z}_{0:N-1}, \\
&x_{(t,0)} = x_t, \quad x_{(t,N)} \in \mathcal{X}_f,
\end{aligned}
\end{equation}
is used to determine the control input $u_t$, where ${\mathbf{u}_t = (u_{(t,0)},\dots,u_{(t,N-1)})}$, and ${\mathbf{x}_t = (x_{(t,0)},\dots , x_{(t,N)})}$. The constraints specify, respectively, the plant dynamics, the state and input constraints, the initial conditions and the terminal constraints. 
The matrices $Q$, $R$, and $P$ in the cost function are user-defined and of appropriate dimensions.
The OCP~\eqref{eq:general_MPC_formulation} is typically solved at every time-step $t$ and its solution is an optimal input and state trajectory, respectively $\mathbf{u}_t^* = (u_{(t,0)}^*,\dots,u_{(t,N-1)}^*)$ and $\mathbf{x}^*_t = (x^*_{(t,0)},\dots , x^*_{(t,N)})$.
The value $N$ is called the \textit{horizon} of the MPC problem and is a parameter that corresponds to how many time-steps are predicted into the future. 
A typical use of the OCP is receding horizon MPC, in which only the first input $u^*_{(t,0)}$ of the input trajectory $\mathbf{u}_t^*$ is applied, i.e., $u_t = u^*_{(t,0)}$ , and then the optimization problem is solved again at time $t+1$.
The terminal set $\mathcal{X}_f$ is chosen as an invariant set to guarantee that the problem remains feasible for all time $t$. The terminal cost term depending on $x_{(t,N)}$ ensures stability \cite{MAYNE2000789}.
As a rule of thumb, a longer horizon allows for more satisfactory performance but leads to higher computational times.

\subsection{Multi-horizon MPC (MH-MPC)}
\label{sec:multi_horizon_mpc_prelim}
We now present the concept of multi-horizon MPC (MH-MPC) \cite{MH_MPC}, which shifts the tradeoff between performance and computational time in a favourable direction when choosing $N$. 
To this end, MH-MPC utilizes models with coarser temporal granularities as the horizon progresses. 
The prediction time is divided into sub-intervals, each one using an increasing sampling time. This reduces the amount of optimization variables and computational load, while considering the same total time interval for prediction, which we call \textit{time-horizon}.

To formalize the above, let us denote $\mathbb{H} = \{1,2,\dots ,\mathcal{H}\}$ as the set of $\mathcal{H}$ sub-intervals. 
The set $\mathbb{K}_i$ denotes the set of steps inside the $i$th sub-interval, and $|\mathbb{K}_i| = h_i$ its cardinality.  
To denote the horizon division we use $H = [h_1,h_2,\dots,h_\mathcal{H}]$. 
Each element of $H$ corresponds to the number of steps taken using a system which propagates the state by a multiple $i$ of the base sample time $T_s$: the first element $h_1$ is the number of steps taken with the base sample time $T_s$, $h_2$ corresponds to the number of steps taken with sample time $2T_s$ and so on (see Fig. \ref{fig:MH-division}).

In each sub-interval, a constant input is applied for $i$ steps on the original system dynamics \eqref{eq:sys} with sampling time $T_s$. 
The dynamics of each sub-interval $i$ can, thus, be constructed for $i=1$ as $A_1=A$ and $B_1=B$ and for $i>1$ as
\begin{equation} \label{eq:time_varying_matrices}
    A_i = A_1^{i}, \quad \text{and}\quad
    B_i = \sum_{j=0}^{i - 1} A_1^j B_1.
\end{equation}

A MH-MPC problem  $\mathcal{P_{MH}}(x_t)$ can be written as follows:
\begin{equation}\label{eq:MH-MPC_problem}
\begin{aligned}
 &\min \sum_{i \in \mathbb{H}} \sum_{k \in \mathbb{K}_i} \Big( \lVert x_{(t,k)}\rVert_{Q_i}^2 +\lVert u_{(t,k)} \rVert_{R_i}^2 \Big)+ \lVert x_{(t,N)}\rVert_{P}^2 \\
&\text{s.t.} \quad x_{(t,k+1)} = A_i x_{(t,k)} + B_i u_{(t,k)}, \quad \forall k \in \mathbb{K}_i, \ i \in \mathbb{H}, \\
&\qquad \ x_{(t,k)} \in \mathcal{X},\ u_{(t,k)} \in \mathcal{U}, \quad \forall k \in \mathbb{Z}_{0:N-1}, \\
&\qquad \ x_{(t,0)} = x_t \in \mathcal{X}, \quad x_{(t,N)} \in \mathcal{X}_f,
\end{aligned}
\end{equation}
where, similar to \cite{MH_MPC}, the cost matrices are defined as
\begin{equation}\label{eq:cost_matrices}    
Q_i = i Q, \quad 
    R_i = i R,
\end{equation}
with $Q$ and $R$ are user-defined as in the standard MPC, which we will denote as uniform-horizon (UH-MPC). 
The MH-MPC is a generalization of UH-MPC, with $H =[N]$  ~\eqref{eq:general_MPC_formulation}.

As the horizon progresses, the matrices $A_i$ and $B_i$ propagate the state holding the input constant for sampling time $iT_s$. Constraints are also enforced only at times $iT_s$. Hence, violations of constraints could occur between ${[}iT_s, (i+1)T_s{]}$. Also, guaranteeing recursive feasibility is a non-trivial task for MH-MPC due to not being able to re-utilize the previously computed input sequence (see \cite{MH_MPC}).

\begin{figure}\
    \centering
    \includegraphics[width=0.9\linewidth]{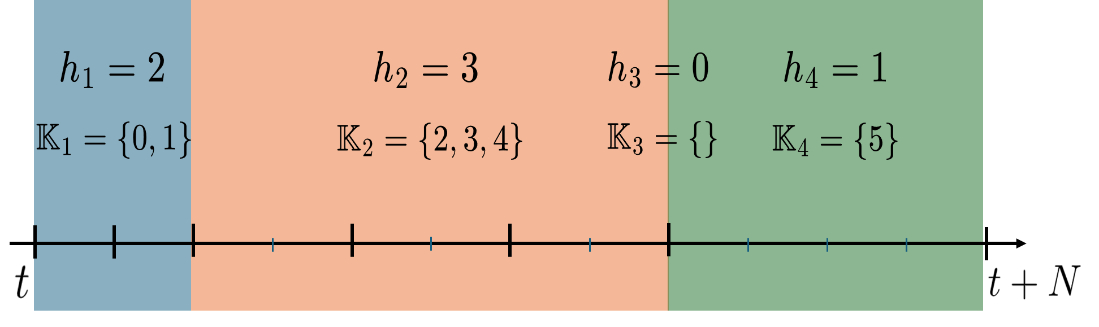}
    \caption{Multi-horizon division example with $H= [2,3,0,1]$}
    \label{fig:MH-division}
\end{figure}

\subsection{Tracking of steady states}
Since we want to design a controller that can track references, we now introduce background information on tracking steady-state values $(\bar{x},\bar{u})$, where $\bar x \in \mathbb{R}^{n_x}$, and $\bar u \in \mathbb{R}^m$. These values are used in MPC approaches for tracking, such as \cite{MPC_for_tracking}. 
The introduction of these variables allows constraint satisfaction and asymptotic evolution of the system to any reference which has an associated admissible steady-state, while also enlarging the domain of attraction of the controller.
If the reference is chosen such that no admissible steady-state pair leads to it, then the system will be steered to the closest admissible steady-state pair \cite{MPC_for_tracking}. 

A steady-state pair $(\bar{x},\bar{u})$ must satisfy
\begin{equation} \label{eq:steady_state_conditions}
    (A-I_{n_x}) \bar x +B \bar u  = \mathbf{0}_{n_x\times1}
\end{equation}
as well as the constraints \eqref{eq:constraints}, i.e, $\bar x \in \mathcal{X}$ and $\bar u \in \mathcal{U}$.

\begin{assumption}\label{ass:stabilizable}
    The pair (A,B) is stabilizable.
\end{assumption}

Assumption \ref{ass:stabilizable} guarantees the existence of a non-trivial solution to \eqref{eq:steady_state_conditions} and it allows us to introduce an ancillary feedback controller $K$ to compute the input $u_t$
\begin{equation}\label{eq:stab_control_law}
    u_t = K(\bar x - x_t) + \bar u, 
\end{equation}
where $K$ is designed such that $A-BK$ is Hurwitz.
We extend the state, in order to construct the autonomous system obtained by applying control law \eqref{eq:stab_control_law}:
\begin{equation} 
    w_t = (x_t, \bar x , \bar u), \quad w_{t+1} =A_e w_t, \label{eq:autonomous_sys}
\end{equation}
where 
\begin{equation} \label{eq: Ae matrix}
    A_e = \begin{bmatrix}
        A-BK & BK & B \\
        \mathbf{0}_{n_x\times n_x} & I_{n_x} & \mathbf{0}_{n_x \times m}\\
        \mathbf{0}_{m\times n_x} & \mathbf{0}_{m \times n_x} & I_{m}
    \end{bmatrix}.
\end{equation}
The constraints under the auxiliary law can be rewritten as $w_t \in W,$
\begin{equation} \label{eq:constraints_autonomous_system}
    W = \lbrace(x,\bar x, \bar u) : x\in\mathcal{X},\ K(\bar x-x)+\bar u)\in \mathcal{U}\rbrace.
\end{equation}

The maximal output admissible invariant set $\mathcal{O}_{\infty}$ \cite{Kolmanovsky} can be computed for the autonomous system \eqref{eq:autonomous_sys} and guarantees that, for any initial state $w_0 \in \mathcal{O}_{\infty}$, the evolution of the system \eqref{eq:autonomous_sys} remains within $\mathcal{O}_{\infty}$ and respects the constraints $W$ at each time step. 
The set is defined as 
\begin{equation} \label{eq:max_admissible_set}
    \mathcal{O}_{\infty} = \{ w= (x,\bar{x}, \bar{u}): A_e^kw\in W,  \ k\geq 0 \}.
\end{equation}
In general $\mathcal{O}_{\infty}$ is not finitely determined, but an approximation $\Tilde{O}_\infty$ can be computed as described in \cite{Remote_MPC,MPC_for_tracking, Kolmanovsky}.
\section{Multi-horizon MPC for tracking}
\label{sec:multi_horizon_mpc}
In the following, we develop a novel MH-MPC for tracking, which is used by the controller in Section~\ref{sec:bw_aware_control} to reduce the size of the control packets sent over the network.
Our design utilizes a new constraint inspired by \cite{Controlled_invariant_feasibility}, which enables us to prove recursive feasibility of our approach. 

Motivated by \cite{MPC_for_tracking}, we define the cost function for our MH-MPC for tracking as
\begin{equation}
\label{eq:MH-MPC_cost}
\begin{split}
V_{MH}(\mathbf{x}_t, \mathbf{u}_t, \bar{x}_t, \bar{u}_t, r_t) = \sum_{i \in \mathbb{H}} \sum_{k \in \mathbb{K}_i} \Big( \lVert x_{(t,k)} - \bar{x}_t \rVert_{Q_i}^2 +\\ + \lVert u_{(t,k)} - \bar{u}_t \rVert_{R_i}^2 \Big) 
+ \lVert x_{(t,N)} - \bar{x}_t \rVert_{P}^2 + \lVert \bar{x}_t - r_t \rVert_{T}^2
\end{split}
\end{equation}
where $P$ and $T$ are appropriately sized matrices, and $Q_i$ and $R_i$ are defined in \eqref{eq:cost_matrices}. 

Using the cost \eqref{eq:MH-MPC_cost}, our MH-MPC problem for tracking, $\mathcal{P}(x_t,r_t)$, is given by
\begin{equation}
\label{eq:MH-MPC_problem}
\begin{aligned}
&\min_{\mathbf{u}_t, \mathbf{x}_t , \bar x, \bar u} V_{MH}(\mathbf{x}_t, \mathbf{u}_t, \bar x_t, \bar u_t, r_t)\\
&\text{s.t.} \quad x_{(t,k+1)} = A_i x_{(t,k)} + B_i u_{(t,k)}, \quad \forall k \in \mathbb{K}_i, \ i \in \mathbb{H}, \\
& \qquad \ x_{(t,k)} \in \mathcal{X},\ u_{(t,k)} \in \mathcal{U}, \quad \forall k \in \mathbb{Z}_{0:N-1}, \\
& \qquad \ x_{(t,0)} = x_t, \quad x_{(t,N)} \in \mathcal{X}  \\
& \qquad \ (x_{(t,h_1)},\bar{x}_t,\bar{u}_t) \in \mathcal{\Tilde{O}}_\infty^{MH}, \quad \bar x_t , \bar u_t  \ \text{satisfy} \ \eqref{eq:steady_state_conditions}
\end{aligned}
\end{equation}

Compared to UH-MPC for tracking \cite{MPC_for_tracking},  MH-MPC \eqref{eq:MH-MPC_problem} does not enforce a constraint on the final state to be able to guarantee recursive feasibility, such as ${(x_{(t,N)},\bar{x}_t,\bar{u}_t)\in \mathcal{O}_\infty}$ in \cite{MPC_for_tracking}.
Instead, we place a constraint on the augmented system at time step $h_1$, i.e., $(x_{(t,h_1)},\bar{x}_t,\bar{u}_t) \in \mathcal{O}_\infty^{MH}$.

To derive the set $\mathcal{O}_{\infty}^{MH}$, we take inspiration from \cite{Controlled_invariant_feasibility}.
Let us first divide the MH-MPC problem defined by the horizon $H=[h_1,h_2, \dots, h_{\mathcal{H}}]$ into two parts.
The first part is defined by $H_1=[h_1]$ and we call the corresponding optimal control problem $\mathcal{P}_1$. 
Note that $\mathcal{P}_1$ is equivalent to UH-MPC with horizon $h_1$.
The second part is $H_2~=~[0,h_2,\dots, h_{\mathcal{H}}]$ with the corresponding problem $\mathcal{P}_2$.
The solutions of these problems are formed by vectors composed of
\begin{align}
\mathcal{P}_1:&
\{\mathbf{x}_t^{(1)}, \mathbf{u}_t^{(1)} \} := \nonumber\\
\{& (x_{(t,0)},\dots , x_{(t,h_1)}), (u_{(t,0)},\dots , u_{(t,h_1-1)}) \} \label{eq:problem_1}
\\
\mathcal{P}_2:& \{\mathbf{x}_t^{(2)}, \mathbf{u}_t^{(2)} \} := \nonumber\\
\{&(x_{(t,h_1)},\dots , x_{(t,N)}), (u_{(t,h_1)}, \dots , u_{(t,N-1)})\} \label{eq:problem_2}
\end{align}
Note that the two problems are coupled through $x_{(t,h_1)}$. 

The set of feasible states for $\mathcal{P}_2$ is given by
\begin{align}
&\mathcal{X}_{0}^{(2)} =\left\{ x_{(t,0)} \in \mathcal{X} \mid \exists \ \mathbf{u}_t^{(2)} \text{s.t.} \ u_{(t,k)}\in\mathcal{U}, \right. \\
&\left. A_i x_{(t,k)} + B_i u_{(t,k)} = x_{(t,k+1)} \in \mathcal{X}, \forall k \in \mathbb{K}_i, \ i \in \mathbb{H}_2 \right\},\nonumber
\end{align}
where $\mathbb{H}_2 = \lbrace 2, 3, \ldots, \mathcal{H}\rbrace$.
With polytopic constraints, $\mathcal{X}_{0}^{(2)}$ can be calculated using precursor sets (see \cite[Ch.10]{bemporad_book}). 

With $\mathcal{X}_{0}^{(2)}$ defined, we consider the autonomous system as described in \eqref{eq:autonomous_sys}, which evolves with the smallest sampling time. 
For the state to remain in $\mathcal{X}_{0}^{(2)}$ when using the controller \eqref{eq:stab_control_law} we introduce the constraint $w=(x, \bar{x}, \bar{u})\in W_{MH}$, with
\begin{align}
\label{eq:constraints_auto_sys_MH}
W_{MH} = \{ (x,\bar x, \bar u) :  
\, x \in \mathcal{X}_{0}^{(2)},\, K(\bar{x} - x) + \bar{u} \in \mathcal{U} \}.
\end{align}
Similar to \eqref{eq:max_admissible_set}, we define $\mathcal{O}_\infty^{MH}$ as 
\begin{equation}  \label{eq:max_admissible_set_MH}
        \mathcal{O}_{\infty}^{MH} = \{ w: A_e^kw\in W_{MH},  \ k\geq 0 \}.
\end{equation}
If $w_t\in \mathcal{O}_{\infty}^{MH}$ and we apply the auxiliary control law \eqref{eq:stab_control_law}, then $w_{t+1} \in W_{MH}$ and, thus, $x_{t+1}\in \mathcal{X}_{0}^{(2)}$.

With $\mathcal{O}_{\infty}^{MH}$ defined, let us now prove the recursive feasibility of \eqref{eq:MH-MPC_problem}.
\begin{proposition}
    The MH-MPC problem $\mathcal{P}(x_t,r_t)$ \eqref{eq:MH-MPC_problem} is feasible at for all times $t\geq t_0$ if it is feasible at $t_0$.
\end{proposition}
 \begin{proof}
    If $\mathcal{P}(x_{t_0},r_{t_0})$ is feasible, then $(x_{(t_0,h_1)},\bar{x}_{t_0},\bar{u}_{t_0})$ $\in \mathcal{O}_\infty^{MH}$. 
    Using the auxiliary control law \eqref{eq:stab_control_law}, the choice ${x_{(t_0+1,h_1)}=(A-BK)x_{(t_0,h_1)}+BK\bar{x}_{t_0}+B\bar{u}_{t_0}}$, guarantees $(x_{(t_0+1,h_1)},\bar{x}_{t_0},\bar{u}_{t_0})\in\mathcal{O}_\infty^{MH}$. 
    Since $x_{(t_0+1,h_1)}\in \mathcal{X}_{0}^{(2)}$, a sequence of feasible inputs for the second part of the problem, $\mathcal{P}_2$, can always be found by the definition of $\mathcal{X}_{0}^{(2)}$. Thus, a feasible solution at time $t_0+1$ exists. Using induction this can be extended $\forall t\geq t_0.$
 \end{proof}

If the horizon is chosen $H = [N]$, the additional constraint becomes the commonly used terminal constraint \cite{Controlled_invariant_feasibility}. 
Hence, the MH-MPC for tracking \eqref{eq:MH-MPC_problem} is a generalization of the UH-MPC formulation for tracking.

\section{Bandwidth-aware control over lossy networks}\label{sec:bw_aware_control}
This section connects the MH-MPC for tracking presented in Section~\ref{sec:multi_horizon_mpc} with the components of the networked control system shown in Fig.~\ref{fig:setup}. A communication rate reduction algorithm is presented, extending the algorithm in \cite{Remote_MPC} with a communication rate parameter $n\in \mathbb{N}$, which denotes the time between packet sends as a multiple of the base sending rate. 
If $n=1$ and the horizon is chosen as uniform, the approach is equivalent to the method presented in $\cite{Remote_MPC}$. 

\subsection{Information sent over the network}
In order to continue operation in the case of packet loss, a trajectory of inputs is sent over the network, allowing the plant to store the received sequence as a buffer. The controller and plant packets are defined as
\begin{equation}
    U_t = \{\mathbf{u}^{*(1)}_t, \bar{x}_t, \bar{u}_t, q_t\} \ , \quad X_t = \{x_t, s_t\},
\end{equation} respectively,
where $\mathbf{u}^{*(1)}_t$, $\bar{x}^*_t$, and $\bar{u}^*_t$ are obtained as \textit{part} of the solution of Problem \eqref{eq:MH-MPC_problem}. Recall that $\mathbf{u}^{*(1)}_t$ represents the first $h_1$ inputs of the optimal input trajectory of \eqref{eq:MH-MPC_problem}. 
Note that an input trajectory of length $h_1$ instead of $N$ is sent, which reduces the bandwidth consumption of our approach compared to \eqref{eq:MH-MPC_problem}.
Variables $q_t \in \mathbb{N}$ and $s_t\in \mathbb{N}$ are discrete time-stamps used by the smart actuator to determine the input to apply to the plant.
The controller packet $U_t$ is sent from time instances $(t-1, t)$, while the plant packet is sent from time $(t,t+1)$.
The \textit{communication-rate parameter} $n \in \mathbb{N_+}$ controls how often each side transmits a packet, expressed as a multiple of the base sample time $T_s$. 

\subsection{Cloud components: Controller \& Estimator}
Every $n$ steps the controller solves and sends the solution to the problem $\mathcal{P}(\hat{x}_{t+1}, r_t)$, where $\hat{x}_{t+1}$ is state estimate at time $t+1$. 
Estimating the state at time $t+1$ is necessary due to delays introduced by communication and computation. 
If the packet from plant to controller has been received successfully at time $t$, i.e. $\gamma_t = 1$, then the state of the plant at time $t$ is available, and the input the plant is applying can be derived from $s_t$.
Hence, we are able to estimate the state perfectly with a one-step ahead prediction. If the packet does not arrive, some assumptions have to be made to estimate the state. 
As in \cite{Remote_MPC}, the estimator assumes that all packets sent from the controller to the plant have arrived successfully - i.e. the communication link is assumed perfect.
We can write the above as:
\begin{align}
    \hat{t} &= t-t\,\mathrm{mod}\,n \\
    \hat{u}_{t|t} &= \gamma_t u_t + (1 - \gamma_t)  u_{(\hat{t},\,t\,\mathrm{mod}\,n)} \\
    \hat{x}_{t|t} &= \gamma_t x_t + (1 - \gamma_t) \hat{x}_{t|t-1} \\ 
    \hat{x}_{t+1|t} &= A\,\hat{x}_{t|t} + B\,\hat{u}_{t|t}.
\end{align}
The variable $\hat{t}$  denotes the instants in which the MPC problem is being computed and the packets sent. 
Since we send packets every $n$ steps, $U_t$ is transmitted when $t \,\mathrm{mod}\, n= 0$. 
The estimator and controller logic is shown in Algorithm \ref{alg:estimator}.

\begin{algorithm}
  \caption{Controller and Estimator algorithm}
  \begin{algorithmic}[1] 
    \State \textbf{Initialize:} $n$
    \For{$t = 0 \to \infty$}
      \State $\hat{t} \gets t-t\,\mathrm{mod}\,n$
      \State $\hat{u}_{t|t} \gets \gamma_t u_t + (1 - \gamma_t)  u_{(\hat{t},\, t\,\mathrm{mod}\,n)}$ \label{alg:input-predict}
      \State $\hat{x}_{t|t} \gets \gamma_t x_t + (1 - \gamma_t) \hat{x}_{t|t-1}$\label{alg:curr_state}
      \State $\hat{x}_{t+1|t} \gets A\,\hat{x}_{t|t} + B\,\hat{u}_{t|t}$ \label{alg:predict}
      \If{$t\,\mathrm{mod}\,n == 0$}
        \State $q_t \gets \gamma_t t + (1 - \gamma_t) q_t$ \label{alg:q_t_update}
        \State \textit{Solve} \eqref{eq:MH-MPC_problem}: $\mathcal{P}(\hat{x}_{t+1|t}, r_t)$ \label{alg:solve}
        \State \textit{Send packet} \label{alg:send}
      \EndIf
    \EndFor
  \end{algorithmic}
   \label{alg:estimator}
\end{algorithm}

\subsection{Local components: Smart Actuator \& Plant}
On the plant side, the \textit{smart actuator} determines the input to apply to the plant.
If a packet does not arrive, the smart actuator resorts to using inputs from the last valid packet received.
If a packet arrives, then a consistency check is performed, to determine if the optimal input sequence has been computed with a correct estimate. 
If the check is successful, the packet is considered valid, else it is discarded.

To determine consistency, a variable $\Theta_t$ is used:
\begin{equation}
\Theta_t =
\begin{cases} 
\prod_{k = 0}^{\lfloor(t-q_t-1)/n\rfloor} \theta_{q_t+1+kn}  & \text{if } \theta_t = 1, \\
0 & \text{otherwise}.
\end{cases}
\end{equation}
Since packets are sent only every $n$ steps, consistency will be obtained whenever the consistency variable $\Theta$ was equal to 1 less than n steps ago. 
\begin{proposition}\label{prop:consistency}
     If \( \Theta_{t-(t\,\mathrm{mod}\,n)}= 1\) then \( x_t = \hat{x}_{t|t-1}\).
\end{proposition}
\begin{proof}
Let $\tau \leq t$ be such that $\gamma_\tau=1$ and $\gamma_{\tau+l}=0$ for $0 \leq l <L=t-\tau$, in other words the last packet has been received by the controller $L$ steps ago. It follows that ${q_t = \tau}$ and $\hat{x}_{t|t-1} = A^Lx_\tau + \sum^{L-1}_{l=0}A^{L-1-l}Bu_{(\tau+ \lfloor\frac{l}{n}\rfloor n, l \,\mathrm{mod}\, n)}$.  $\Theta_t =1$ by definition implies that all the packets sent by the controller since time $\tau$ have been received and used. This means $x_t =A^Lx_\tau + \sum^{L-1}_{l=0}A^{L-1-l}Bu_{(\tau+ \lfloor\frac{l}{n}\rfloor n, l \,\mathrm{mod}\, n)}$. Hence, $\Theta_t = 1$ implies $x_t = \hat{x}_{t|t-1}$. 
\end{proof}

The variable $s_t$, representing the time at which the last consistent packet was received, is updated with ${s_{t+1} = \Theta_t t + (1 - \Theta_t) s_t}$. Note that $s_t$ changes at most every $n$ steps, since $\Theta_t = 0$ when $t \,\mathrm{mod}\,  n\neq 0.$ The smart actuator law becomes:
\begin{equation} \label{eq:smart_actuator_law}
    u_t =
\begin{cases} 
u_{(s_t,t-s_t)} & \text{if } t - s_t < h_1, \\
\bar{u}_{s_t} + K(x_t - \bar{x}_{s_t}) & \text{otherwise}.
\end{cases}
\end{equation}

If the plant runs out of buffered inputs it resorts to applying an auxiliary control law, from which the set $\mathcal{O}^{MH}_\infty$ is also defined. 
This has the same bandwidth requirements as using as base sampling time $nT_s$ but allows for better performance with respect to controlling the system with an increased sample time $nT_s$ since the input can change every $T_s$. 
This is validated experimentally in Section \ref{sec:Experiments}.
Finally, the smart actuator transmits $X_t$ when $(t-1) \,\mathrm{mod}\, n~=~0$.
\subsection{Theoretical properties}
We now formalize how constraint satisfaction and recursive feasibility are guaranteed over an arbitrary network.
\begin{theorem}\label{prop:recursive_feasibility}
    Given Assumption \ref{ass:network}, there exists a time instant $t_0 \ s.t. \ t_0 \,\mathrm{mod}\,n =0$ where a successful back to back transmission occurs, i.e.  $\gamma_{t_0-1} = 1$, $\theta_{t_0} = 1$. Assume the optimization problem $\mathcal{P}(x_{t_0}, r_{t_0})$ defined as in \eqref{eq:MH-MPC_problem} is feasible. Then, if controller and estimator logic is chosen as in Algorithm \ref{alg:estimator}, and the smart actuator law as in \eqref{eq:smart_actuator_law}, then $\mathcal{P}(\hat{x}_{t|t-1}, r_{t})$ is feasible and $x_t\in \mathcal{X},u_t \in \mathcal{U}, \ \forall t\geq t_0$.
\end{theorem}
\begin{proof}
By induction. The problem $\mathcal{P}(x_{t_0}, r_{t_0})$ is feasible at time $t_0$ by assumption. Assume that the problem is feasible for $\tau <t$. We have three distinct cases: $\gamma_t= 0$, $\gamma_t =1$ and $\Theta_{t-1} = 1$, $\gamma_t = 1$ and $\Theta_{t-1}~=~0$. When $\gamma_t = 0 $, then we can reuse the previous $h_1-1$ inputs from the previous shifted sequence and always find an input to append to it due to the constraint $(x_{(t,h_1)}, \bar x_t, \bar u_t) \in \mathcal{O}_\infty^{MH}$, hence the constraints related to $\mathcal{P}_1$ are fulfilled.  
Then, the constraints related to $\mathcal{P}_2$ can be fulfilled since $Proj_x(\mathcal{O}_\infty^{MH})\subseteq X_0^{(2)}$, so we can always find an admissible input sequence that respects the constraints. If $\gamma_t = 1$ and $\Theta_{t-1} = 1$, then by Proposition \ref{prop:consistency} we have $x_t = \hat{x}_{t|t-1}$, so we can use the same arguments as in the former case. When $\gamma_t = 1$ and $\Theta_t = 0$, then there exists a $\tau<t$ such that $\Theta_{t+l} = 0$, for $0<l \leq t-\tau=L$ and $\Theta_\tau=1$, so that $x_{\tau} = \hat{x}_{\tau|\tau-1}$. By assumption of the inductive argument, the problem is feasible at time $\tau$ and provides the sequence $\mathbf{u}_{\tau}$. Then, when $L<\hat{N}$ we can generate a admissible sequence by using the remaining $\hat{N}-L$ inputs from the sequence $\mathbf{u}_{\tau}$ and $L$ inputs from the auxiliary control law. Then, the remaining inputs can always be found since $X_0^{(2)} \in \mathcal{O}_\infty^{MH}$. If $L> \hat{N}$, then $\hat{N}$ inputs can be found from the auxiliary control law and the rest will be found for the aforementioned reason. 
\end{proof}

We now introduce some additional assumptions, commonly used in MPC.
\begin{assumption}\label{ass:MPC_assumptions}
    The following conditions hold: 
    \begin{itemize}
        \item $Q,R,T$ are positive definite;
        \item $K$ is a constant stabilizing gain for system \eqref{eq:sys};
        \item $P = (A+BK)^TP(A+BK) + Q + K^TRK$.
    \end{itemize}
\end{assumption}

Now we can show that the proposed approach also allows for reference tracking.
\begin{theorem}\label{prop:reference_tracking}
    Suppose that Assumptions \ref{ass:MPC_assumptions} and \ref{ass:network} hold, and the initial problem is feasible. Assume a uniform horizon $H = [N]$. Let r be such that $\bar x \in \mathcal{X}, \ \bar u \in \mathcal{U}$. If choosing the controller and estimator logic as shown in Algorithm \ref{alg:estimator}, and the smart actuator law as \eqref{eq:smart_actuator_law}, then $\lim_{t\to\infty} x_t \stackrel{a.s.}{=}r$.
\end{theorem}
\begin{proof}
Thanks to Proposition \ref{prop:consistency}, and the fact that the variables $\theta_t$, $\gamma_t$ are updated at every time step we can re-use the proof from Proposition 3 of \cite{Remote_MPC}. 
\end{proof}

\section{Experiments} \label{sec:Experiments}
We validate the proposed bandwidth-aware MPC algorithms in a real-time 5G-in-the-loop setup. 

Experiments are conducted over a private 5G network provided by Ericsson Research \cite{KIP}. 
The controller is executed on a remote server, while the real-time plant simulation runs on a local workstation. The two systems are clock synchronized.
The plant is a nonlinear cart-pole system, implemented with PyBullet physics engine. The simulation runs at 200Hz to accurately capture the system dynamics. The system and its parameters, as well as the cost matrices for the controllers are chosen as described in \cite{Tube_MPC}.

The controller is implemented in Python using CVXPY \cite{cvxpy} and solved with interior-point solver CLARABEL \cite{Clarabel_2024}. 
The model of the system is obtained by linearizing around the upright equilibrium using a zero-order hold discretization with the sample time used for control. 

Three sets of experiments are conducted for evaluation: 1) Compare the communication rate reduction with simply increasing the sampling time. 2) Compare MH-MPC with UH-MPC. 3) Conduct a network load analysis over 5G network.

\subsection{Communication-rate reduction vs. increased sampling time}
We first present results with a fixed horizon and varying the \textit{communication rate parameter} $n$ algorithm introduced in Section \ref{sec:bw_aware_control}. As baseline we set $n=1$. We compare a naive strategy that increases the sampling time to $2T_s$ against using base sample time and $n=2$. The system is linearized to match the sampling time to ensure model consistency. All controllers use a uniform horizon of $H=[30]$.

\begin{figure}
\centering
\includegraphics[width=\linewidth]{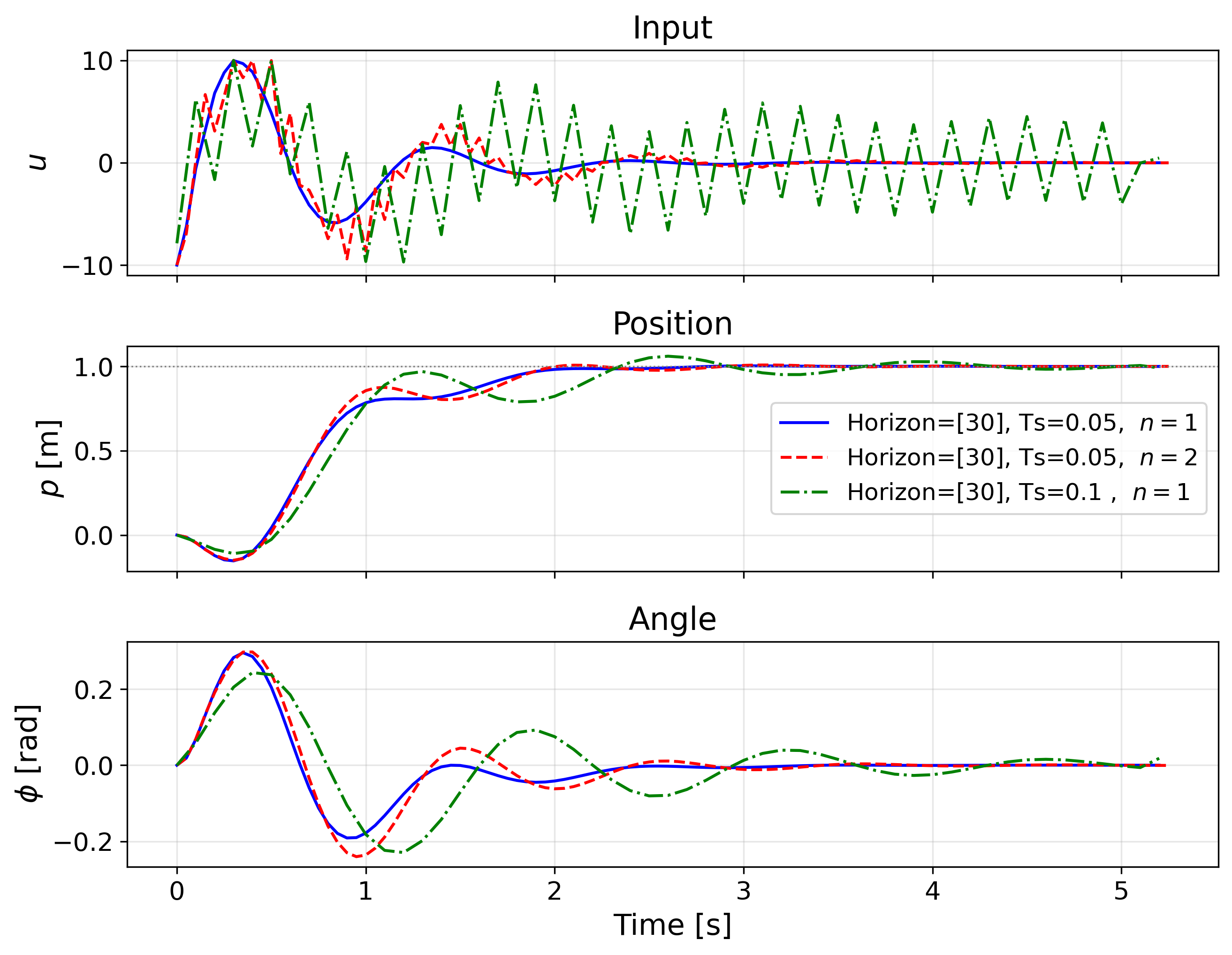}
\caption{Trajectory comparison of standard MPC with different sampling time vs communication rate reduction method.}
\label{fig:sending_rate_comparison}
\end{figure}

As shown in Fig.~\ref{fig:sending_rate_comparison}, introducing the communication rate parameter $n$ has a similar performance as the baseline {($n=1$)} while reducing the communication frequency by a factor of two. In contrast, increasing $T_s$ results in a clear degradation of control performance, confirming the advantages of the rate-reduction strategy.

\subsection{Multi-horizon MPC vs. uniform-horizon MPC} We evaluate performance of the MH-MPC formulation. The discretization $H = [5,4,3,2]$ results in 14 control steps and is compared with two UH-MPC baselines:
(i) a short-horizon controller with $H=[5]$ (same packet size), and
(ii) a long-horizon controller with $H=[30]$ (same prediction length in physical time).
Fig.~\ref{fig:multiple_trajectoris} shows closed-loop trajectories while tracking a constant reference. The MH-MPC is able to bring the system to the desired reference with a performance comparable to the UH-MPC choice. The short UH-MPC ($H=[5]$) exhibits the worst control performance and is not able to stabilize the plant. 
Note that the discontinuities in the input sequence are caused by packet loss over the network.
The computational times for the different MPCs are shown in Table \ref{table:CPT_comparison}. On average, the computational time of MH-MPC is approximately $5\,\text{ms}$ lower than for UH-MPC with $H =[30]$. However, MH-MPC cannot be solved as quickly as UH-MPC with $H=[5]$.

\begin{figure}
\centering
\includegraphics[width=\linewidth]{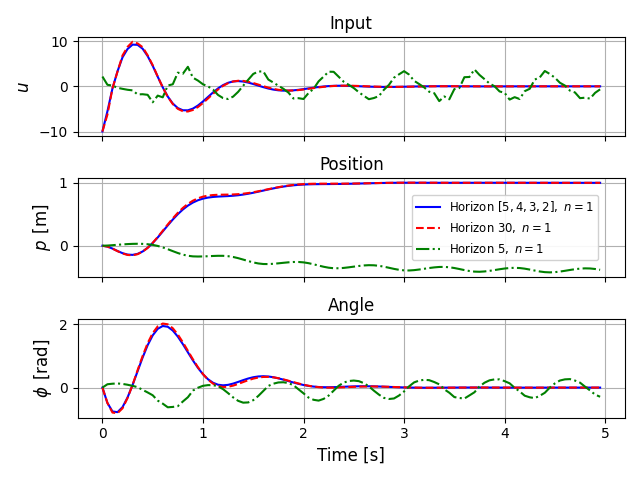}
\caption{Trajectories with different horizon choices}
\label{fig:multiple_trajectoris}
\end{figure}

\begin{table}
\centering
\caption{Computational Time Comparison}
\label{table:CPT_comparison}
\renewcommand{\arraystretch}{1.15}
\begin{tabular}{@{} l S[table-format=2.2] S[table-format=1.2] @{}}
\toprule
Horizon & {CPT mean (\si{\milli\second})} & {Std deviation (\si{\milli\second})} \\
\midrule
{[5,4,3,2]} & 11.36 & 1.07 \\
30 & 16.54 & 2.90 \\
5 & 9.95 & 2.01 \\
\bottomrule
\end{tabular}
\end{table}

\begin{table*}
\centering
\caption{Uplink Congestion - 8 Mbps}
\label{tab:experiment-results}
\renewcommand{\arraystretch}{1.2}
\begin{tabular}{
  c % Sending rate (kept simple to work well with \multirow)
  c % Horizon
  S[table-format=2.0] % Experiments successful
  S[table-format=2.2(2.2), separate-uncertainty] % Uplink loss mean (%)
  S[table-format=2.2(2.2), separate-uncertainty] % Downlink loss mean (%)
  S[table-format=2.2(2.2), separate-uncertainty] % Avg CPT time (ms)
  S[table-format=1.2(1.2), separate-uncertainty] % MSE mean
}
\toprule
Comunication rate & Horizon & {Experiments successful} & {Uplink loss mean (\%)} & {Downlink loss mean (\%)} & {Avg CPT (ms)} & {MSE mean} \\
\midrule
\multirow{2}{*}{1} & {[5,4,3,2]} & 20 & 41.54 \pm 14.62 & 15.23 \pm 4.22  & 11.38 \pm 2.30 & 0.17 \pm 0.00 \\
                   & 30           & 20 & 31.86 \pm 20.37 & 36.21 \pm 16.33 & 16.92 \pm 3.19 & 0.19 \pm 0.01 \\
\midrule
\multirow{2}{*}{2} & {[5,4,3,2]} & 13 & 0.85  \pm 2.26  & 0.65  \pm 2.13  & 11.05 \pm 2.31 & 0.18 \pm 0.00 \\
                   & 30           & 12 & 1.39  \pm 3.06  & 8.43  \pm 16.69 & 17.04 \pm 3.5  & 0.81 \pm 1.42 \\
\bottomrule
\end{tabular}
\label{table:congestion}
\end{table*}

\subsection{Network-load analysis}
We now run multiple controllers and plants simultaneously. This results in higher network load and makes the effect of the bandwidth reduction parameters more visible. In the following, $20$ controller-plant pairs are run simultaneously.

\begin{figure}
  \centering
  \includegraphics[width=0.75\linewidth]{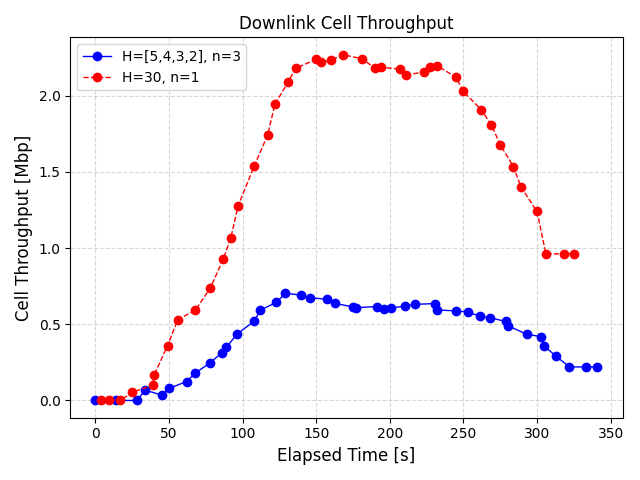}\\[-0.5em]
  \includegraphics[width=0.75\linewidth]{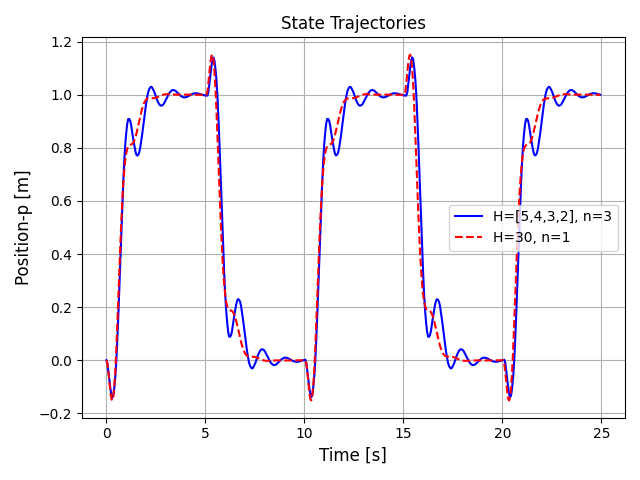}
  \caption{Data represents 20 plant-controller pairs executed simultaneously. The initial and final slope in the throughput reflects the startup time required to initialize all plant-controller pairs. }
  \label{fig:BB_data}
\end{figure}

Using 5G base-band telemetry, we monitor cell throughput.
The following choice of horizon/communication rate combinations was used: $H=30, \ n=1$ and ${H=[5,4,3,2], \ n=3}$. The resulting network metrics are shown in Fig.~\ref{fig:BB_data} along with the corresponding trajectories.
We observe that the bandwidth is approximately reduced by a factor of $n=3$, while tracking performance is only marginally affected. 
The packet length has no discernible impact on bandwidth consumption because the payload is small relative to the size of the packet headers. 
With communication protocols that minimize header overhead, packet length may become more critical for the bandwidth.

To emulate bandwidth-intensive workloads, we saturate the robot-to-server (uplink) channel with 8 Mbps of UDP traffic. This can represent continuous onboard video streaming from camera sensors. This load stresses the communication path, increasing end-to-end latency, causing packet loss. 
Notably, uplink congestion also increases the effective packet loss observed on the server-to-robot (downlink) connection. 
The mechanism is indirect: higher uplink delays postpone the controller’s receipt of plant measurements, delaying the start of the optimization problem and shrinking the time budget available to compute and transmit the corresponding control packet. 
Packets arriving at the plant after their deadline are treated as lost, increasing the downlink loss rate. 

We conducted twenty experiments under two horizon configurations: a MH-MPC set $H=[5,4,3,2],\  n \in \{1,2\}$ and UH-MPC with $H=[30], \ n \in \{1,2\}$. 
We define experiment failure as the inability of the controller to track the reference or stabilize the plants.
The results are summarized in Table \ref{table:congestion}. 
As shown in the third column, increasing the communication rate parameter $n$ reduces the success rate.
This is due to the system operating in open loop for a long periods of time, which can exacerbate the effect of model mismatch. 
For both horizon settings, reducing the communication rate markedly lowers the packet-loss percentage. The MH formulation achieves the lowest downlink packet loss, consistent with its shorter computation times reported in column 6; the reduced solve time increases the slack for timely downlink delivery before deadlines. The MH case also yields lower MSE, attributable to earlier activation of the local controller enabled by smaller packet sizes, thus limiting the duration of open-loop operation.

\section{Conclusions}
 We have shown:
(i) the proposed communication rate reduction strategy achieves better control performance than naively increasing the sampling time;
(ii) the proposed multi-horizon MPC formulation enhances performance while reducing computational effort compared to a standard MPC with uniform discretization and 
(iii) the strategy enables a significant down-link bandwidth reduction over 5G;

We have shown the critical role of computational time in networked control systems and proposed methods to reduce both bandwidth usage and computation times without compromising control performance. Future work could focus on leveraging the time gained from reduced communication rates, for instance, by pre-computing control solutions. Another interesting aspect to be investigated could be the development of adaptive algorithms, capable of tuning system parameters online in response to varying network conditions. 

\section*{Acknowledgment}
The authors would like to thank Luca Schenato for taking the first step at initiating this fruitful collaboration, and well as for insights and timely feedback. The AI model GPT-5 was used to polish the text, all output was carefully reviewed.

\bibliography{ref}

\end{document}